\documentclass[12pt,reqno,fleqn]{amsart}
\usepackage{latexsym}
\usepackage{amssymb}
\usepackage{amsxtra}
\usepackage{amsmath}
\usepackage{amsfonts}
\usepackage{CJK}
\usepackage{cite}

\newcommand{\cleqn}{\setcounter{equation}{0}}
\newcommand{\clth}{\setcounter{theorem}{0}}
\newcommand {\sectionnew}[1]{\section{#1}\cleqn\clth}
\newtheorem{theorem}{Theorem}[section]
\newtheorem{lemma}[theorem]{Lemma}

\newtheorem{definition}[theorem]{Definition}

\newcommand{\p}{\mathbf{p}}

\def\({\left(}
\def\){\right)}
\def\[{\left[}
\def\]{\right]}

\begin{document}

\title[ncKP Hierarchy]{Explicit Flow Equations and Recursion Operator of the ncKP hierarchy}
\author{Jingsong He $^1$
, Junyi Tu $^2$, Xiaodong Li $^2$, Lihong Wang$^1$ }
\dedicatory {$^1$Department of Mathematics, NBU, Ningbo, 315211 Zhejiang, P.\ R.\ China \\
$^2$Department of Mathematics, USTC, Hefei, 230026 Anhui, P.\ R.\
China }

\thanks{$^*$ Corresponding author: hejingsong@nbu.edu.cn,jshe@ustc.edu.cn}
\texttt{}

\date{}

\begin{abstract}
The explicit expression of the flow equations of the noncommutative
Kadomtsev-Petviashvili(ncKP) hierarchy is derived. Compared with the flow equations of the KP hierarchy,
our result shows that the additional terms in the flow equations of the ncKP hierarchy indeed
consist of commutators of dynamical coordinates \{$u_i$\}. The recursion operator for the flow equations under
$n$-reduction is presented. Further, under $2$-reduction, we calculate a nonlocal recursion operator $\Phi(2)$
 of the noncommutative Korteweg-de Vries(ncKdV) hierarchy, which generates
a hierarchy of local, higher-order flows. Thus we solve the open problem proposed by
P.J. Olver and V.V. Sokolov(Commun.Math.Phys. 193 (1998), 245-268).
\end{abstract}

\maketitle

{\bf AMS classification(2010):} 35Q53,37K05,46L55

{\bf Keywords:}  ncKP hierarchy, flow equations, recursion operator, ncKdV hierarchy.\\
\allowdisplaybreaks

\sectionnew{Introduction}

The noncommutative integrable system is one of most important and attractive topics in mathematical physics. Some
early interesting results were contributed by P.D. Lax\cite{lax},
M. Wadati and K. Tamijo\cite{wadati1974}, F. Calogero and A. Degasperis
\cite{calogero} and other authors. At that time, matrix soliton equations,
such as  matrix KdV equation and matrix nonlinear Schr\"odinger
equation, were typical noncommutative systems.  The noncommutative KP
hierarchy(ncKP), since it was summarized by P. Etingtof, I.M. Gelfand
and V.S. Retakh \cite{Gelfand1}, has been studied extensively during the last decade
\cite{Gelfand2,Takasaki,olver,fokas,Kupershmid,Paniak,zheng1,zheng,hamanaka1,
hamanaka2,Hoissen1,Hoissen2,wadati,Treves,Nimmo,Nimmo2}. Many authors have discussed the ncKP hierarchy from different viewpoints, such as a generalized ncKP hierarchy\cite{fokas} derived by means of Hamiltonian theory over noncommutative rings, Lagrangian and Hamiltonian formalism \cite{Kupershmid}, the exact multi-soliton
solutions\cite{Paniak}, the B\"acklund transformation of the
noncommutative Gelfand-Dickey hierarchy\cite{zheng1}, the existence of infinite conserved laws\cite{hamanaka2}, the algebraic property\cite{Treves}, quasi-determinant solutions of the noncommutative modified KP equations generated by Darboux transformation\cite{Nimmo,Nimmo2}. In particular,  the supersymmetric KP hierarchy, which is one special ncKP hierarchy, has been studied from perspectives of Lax equations, Hamiltonian structures and the solvability by Y. Manin and M. Mulase\cite{manin,mulase}.

Up to now, research results show that most of the integrable
properties, such as Lax equations, the commutative property of the flows, soliton solutions, Hamiltonian structures, the existence of infinite many conserved laws, can be generalized to the ncKP hierarchy\cite{hamanaka1}. But these generalized results do not mean that the ncKP hierarchy is a trivial formalism generalization of the KP hierarchy. Indeed, the concept of quasi-determinant was crucial to express the exact multi-soliton solutions of ncKP hierarchy \cite{Gelfand2}. Compared with the KP hierarchy, one can expect intuitively that there may be some additional terms consisting of commutators in the flow equations of the ncKP hierarchy. But it is not easy to determine the concrete form of commutators appearing in the flow equations. In \cite{Hoissen2}, some flow equations of Moyal-deformation KdV hierarchy are derived, but the results cannot lead to a general conclusion on the flow equations of the ncKP hierarchy. What is worse, the existence or nonexistence of $\tau$ function for ncKP hierarchy is still an open problem
\cite{Kupershmid}, see page 220.

In view of the pivotal position of the recursion
operator\cite{olver1977,olverbook,magri,fokas1,fokas2,fokas3,fokas4,sokolov,Oevel} in the theory of integrable system, it is very natural to extract recursion operator of the ncKP hierarchy under $n$-reduction from the explicit
flow equations.  This idea is a noncommutative analogue of the corresponding method for the KP hierarchy given by W. Strampp and W. Oevel \cite{Oevel}. The recursion operator of the $n$-reduction KP hierarchy is essential for obtaining the higher-order flows and symmetries starting from the lower-order ones. The construction of recursion operator in 2+1 dimension was the most important open problem in
the algebraic theory of integrable PDEs in the early 80s. This problem was solved in the papers by Fokas and Santini
\cite{fokas1,fokas2,fokas3}. A comprehensive review can be found in \cite{fokas4}. There are several ways to construct the recursion operator of a given integrable system. V.V. Sokolov et al. provided a comprehensive survey about constructions of recursion operator
in \cite{sokolov}. They also presented a way to
construct the recursion operator of the Gelfand-Dickey hierarchy from Lax representation with the help of one ansatz $\tilde{A}=PA+R$. According to the classification in \cite{sokolov}, W. Strampp-W. Oevel method belongs to the methodology of using Lax representation although it was not cited in \cite{sokolov}.
The advantage of this method is that the higher-order flows generated by recursion operator from lower-order ones are local even if the recursion operator has nonlocal term, because the higher-order flows are
automatically identified with the local flows given by the Lax equations.

 Another efficient way to construct recursion operator is to determine the Hamiltonian
 operator $\theta_{1}$ and $\theta_{2}$ of the equation under consideration, then the recursion
 operator is given by $\Re=\theta_{2}\theta_{1}^{-1}$ \cite{magri}. This approach was applied
 to the noncommutative Korteweg-de Vries(ncKdV) equation by P.J. Olver and V.V. Sokolov in \cite{olver}. They obtained a nonlocal recursion operator, and then
proposed an open problem that their recursion operator produces a hierarchy of local, higher-order flows, which is still not proved in literatures to the best of our knowledge.  I.Y. Dorfman and A.S. Fokas also derived a nonlocal recursion operator (see Eq.(28) of \cite{fokas}) to
generate  a generalized ncKP hierarchy. Recently, S. Carillo and C. Schiebold\cite{Carillo1} used this recursion operator to obtain the noncommutative potential KdV hierarchy and derived its ``operator soliton''. So it is necessary to give a positive answer to this open problem in order to ensure the local property of the flows generated by
the nonlocal recursion operator.

The purposes of this paper are twofold. 
The first is to prove the conjecture on the additional
commutators in the flow equations of the ncKP hierarchy by presenting the explicit expression of the flow equations. The second is to provide an affirmative solution to the open problem proposed by P.J. Olver and V.V. Sokolov. To this end we shall first calculate the recursion operator of the ncKdV
hierarchy from the general formula of the $n$-reduction ncKP
hierarchy, then after a rescaling we recover the
recursion operator obtained by P.J. Olver and V.V. Sokolov in \cite{olver}. Since we have known that the higher-order flows of the ncKdV hierarchy are local from its explicit expression before obtaining the recursion operator, there is no doubt that the recursion operator produces
a hierarchy of local flows. There are two key steps in our calculation, one is to keep the noncommutative coefficients on the left side of $\partial$, and the other is to
introduce a left-multiplication operator in the explicit expression
of the flow equations. 

The paper is organized as follows. In Section 2 we review the definition of Lax equations of the ncKP hierarchy
within the framework of Sato's theory\cite{sato}. In particular we give a terse formulation for the corresponding flow equations. The notion of $n$-reduction is also considered. Section 3 is devoted to the explicit expression of the flow equations of the  ncKP hierarchy, which are local flows obviously. In particular, the concrete forms of additional commutators on noncommutative dynamical
coordinates \{$u_i$\} are given. We also calculate some higher-order flow equations via the computer algebra. As examples, some noncommutative analogue of the well known  integrable equations including the KP
and KdV are derived. In Section 4 we present a recursion operator which maps the lower-order flow equations to higher-order flow
equations in the general assumption of $n$-reduction. The ncKdV and
fifth-order ncKdV equations are derived again by the recursion operator of the ncKdV hierarchy, which coincide with the equations given by
2-reduction in Section 3. By a rescaling, we
recover the recursion operator of the ncKdV hierarchy in \cite{olver}, then elucidate how the open problem proposed by P.J. Olver and V.V. Sokolov\cite{olver} is solved. Section 5 encompasses some conclusions and
discussions.

\sectionnew{ The KP hierarchy with noncommutative coefficients }
   In this section we study the KP hierarchy with noncommutative coefficients and show how the classical KP
hierarchy can be generalized to the ncKP hierarchy\cite{Gelfand1,Kupershmid}. The sub-hierarchy of $n$-reduction
and a terse formulation of the corresponding flow equations will be introduced.

     Let $(R,\partial)$ be a differential algebra defined over a field
  $k$ of characteristic 0 with the unity 1. Namely, $R$ is an
  associative $k$-algebra with unity and $\partial: R\rightarrow R$
  is a $k$-linear map satisfying the Leibnitz rule $\partial (fg)=\partial (f) g+f\partial
  (g)$ for all $f$ and $g$ in $R$. We do not assume commutativity of
  $R$. We also use the subscript notation for $\partial $ (regarded as derivative), for instant,
$u_x,u_{xx},\cdots $ to denote the first, second derivatives of
$u=u(x)$ and so on with respect to its scalar variable $x$. The commutator is
denoted by the standard bracket notation $$[u,v]=uv-vu.$$

   Following the classical KP hierarchy, we introduce a Lax operator as
   the first-order pseudo-differential operator:
\begin{eqnarray}L=\sum_{l\leq 1} u_{1-l}\partial^{l}=\partial+u_{2}\partial^{-1}+u_{3}\partial^{-2}+\cdots
,
\end{eqnarray}
where   $u_i\in R,i\in \mathbb{N},u_{0}=1,u_{1}=0$.
 $\partial^{-1} $ is a formal integration operator, satisfying
$\partial \partial^{-1}=\partial^{-1}\partial=1.$ The operation of
$\partial^{\nu}$ with $\nu \in \mathbb{Z}$ is given by the general Leibnitz
rule
\begin{eqnarray}
 \partial^{\nu}f=\sum_{j\geq 0} \binom{\nu}{j} f^{(j)} \partial^{\nu -j},
\end{eqnarray}
where $f\in R, f^{(j)}=\partial^{j} (f)$ and the binomial
coefficients $ \binom{\nu}{j}$ are defined
$$
\binom{\nu}{j}=\dfrac{\nu (\nu-1)(\nu-2)\cdots (\nu-(j-1))}{j!},j\geq 0.
$$


    We  denote
\begin{eqnarray}
L^{m} &=& (\partial+u_{2}\partial^{-1}+u_{3}\partial^{-2}+\cdots)^{m}\nonumber\\
      &=& \sum_{j\leq m} p_{j}(m)\partial^{j},
\end{eqnarray}
and
\begin{eqnarray}
 & &L_{+}^{m} = \sum_{j=0}^m
p_{j}(m)\partial^{j},\quad L_{-}^{m} =
\sum_{j< 0} p_{j}(m)\partial^{j}. \nonumber\\
& &L_{+}^{m} = \sum_{j=0}^m
\partial^{j}q_{j}(m),\quad L_{-}^{m} =
\sum_{j< 0} \partial^{j} q_{j}(m). \nonumber
\end{eqnarray}
It is easy to find that $p_{m}(m)=1,p_{m-1}(m)=0 $.
    For any fixed $m$, there is a one-to-one
correspondence between any two of the three infinite sets of
coordinates: ($u_2,u_3,\cdots $), ($p_{m-2}(m),p_{m-3}(m),\cdots $) and ( $q_{m-2}(m),q_{m-3}(m),\cdots $ ).
We usually call the first set ($u_2,u_3,\cdots $) dynamical coordinates of the ncKP hierarchy.


\begin{definition}The noncommutative KP hierarchy is defined by following Lax equations
\begin{eqnarray} \label{Laxeq}
 L_{t_{m}}=[L_{+}^{m},L],
\end{eqnarray}
which are infinite numbers of partial differential equations on dynamical coordinates $u_i(i\geq 2)$
with respect to infinite numbers of time variables $t=(t_1,t_2,t_3,\cdots)$.
\end{definition}
We have the following theorem for the specific flow equations.
\begin{theorem}
The flow equations of the ncKP hierarchy can be expressed as
\begin{eqnarray} \label{floweqwithp}
u_{j,t_{m}}=\sum_{h=1}^j A_{j,h}p_{-h}(m), j=2,3,4,\cdots
\end{eqnarray}
where 
\begin{eqnarray}
A_{j,h}=\sum_{r=0}^{j-h}\Big(\binom{1-r}{j-h-r}u_{r}\partial^{(j-h-r)}-
                           \binom{-h}{j-h-r}
                          \stackrel{\longleftarrow}{u_{r}^{(j-h-r)}}\Big)\\
 j=2,3,4,\cdots,1\leq h \leq j\;.\nonumber
 \end{eqnarray}
 Here the left-multiplication operator ``$\stackrel{\longleftarrow}{}$" is defined
as $\stackrel{\longleftarrow}{g} f=fg $.

\end{theorem}

\begin{proof}
 Since $$L_{t_{m}}=[L_{+}^{m},L]=[L,L_{-}^{m}],$$ substituting
$ L_{-}^{m} =\sum_{h> 0} p_{-h}(m)\partial^{-h}$ into the equation
above yields
\begin{eqnarray}
&&\mbox{\hspace{-0.4cm}}LL_{-}^{m}-L_{-}^{m}L\nonumber \\
&\mbox{\hspace{-0.4cm}}=\mbox{\hspace{-0.5cm}}&\sum_{r\geq 0}
u_{r}\partial^{1-r}\sum_{h>
                           0}p_{-h}(m)\partial^{-h}-\sum_{h> 0}p_{-h}(m)
                              \partial^{-h}\sum_{r\geq 0} u_{r}\partial^{1-r}\nonumber\\
                      &\mbox{\hspace{-0.8cm}}=\mbox{\hspace{-0.6cm}}&\sum_{r\geq 0}\sum_{h> 0}\sum_{\alpha\geq 0}\Big(\binom{1-r}{\alpha} u_{r}p_{-h}^{(\alpha)}(m)-
                           \binom{-h}{\alpha} p_{-h}(m)u_{r}^{(\alpha)}\Big)\partial^{1-r-h-\alpha} \nonumber\\
                      &\mbox{\hspace{-0.8cm}}=\mbox{\hspace{-0.6cm}}&\sum_{k\geq 0}\sum_{r= 0}^k\sum_{h>0}\Big(\binom{1-r}{k-r} u_{r}p_{-h}^{(k-r)}(m)-
                           \binom{-h}{k-r} p_{-h}(m)u_{r}^{(k-r)}\Big)\partial^{1-h-k}\nonumber\\
                      &\mbox{\hspace{-2cm}}=\mbox{\hspace{-2.4cm}}&\sum_{j\geq 1}\sum_{h=1 }^j\sum_{r= 0}^{j-h}\Big(\binom{1-r}{j-h-r} u_{r}p_{-h}^{(j-h-r)}(m)-
                           \binom{-h}{j-h-r}
                           p_{-h}(m)u_{r}^{(j-h-r)}\Big)\partial^{1-j}\nonumber \\
                      &\mbox{\hspace{-2cm}}=\mbox{\hspace{-2.4cm}}&\sum_{j\geq 1} A_{j,h}p_{-h}(m)\partial^{1-j} .\nonumber
\end{eqnarray}
Note that the summation above starts actually from $2$, because $A_{j,h}|_{j=1}=A_{1,1}=0$.
By comparing  with
\begin{eqnarray*}
L_{t_{m}}=\sum_{j\geq 2} u_{j,t_{m}}\partial^{1-j},
\end{eqnarray*}
we obtain (\ref{floweqwithp}).
\end{proof}

 Introducing left-multiplication operator $``\longleftarrow"$ is the key to derive the general
expression (\ref{floweqwithp}) as the flow equations of the ncKP hierarchy. It is also the striking difference from the
classical KP hierarchy. At present it seems to be not necessary to introduce such a weird operator, but in
Section 4 we shall see the advantage of such expression in order to obtain recursion operator.
We calculate some examples of  $A_{j,h}$,
\begin{eqnarray*}
A_{j,j}&=& 0\;,\\A_{j,j-1}&=& \partial\;,\\A_{j,j-2}&=&u_{2}-\stackrel{\leftarrow}{u_{2}}\;,\\
A_{j,j-3}&=& -u_{2}\partial
-(3-j)\stackrel{\leftarrow}{u_{2,x}}+u_{3}-\stackrel{\leftarrow}{u_{3}}\;
.\nonumber
\end{eqnarray*}

Since $p_{j}(m)$'s are uniquely
determined as functions of $u_2,u_3,\cdots ,u_{m-j}$, precisely
\begin{eqnarray}\label{diffploy}
p_{j}(m)=mu_{m-j}+f_{jm}(u_2,u_3,\cdots ,u_{m-j-1})
\end{eqnarray}
with certain differential polynomials $f_{jm}$ in $u_2,u_3,\cdots
,u_{m-j-1}$. Inserting these into (\ref{floweqwithp}), we obtain the flow
equations for the infinite numbers of dynamical coordinates ($u_2,u_3,\cdots$).
Taking the examples of  $A_{j,h}$ above into
(\ref{floweqwithp}), we have the first three flow equations:
\begin{eqnarray*}
u_{2,t_{m}} &=& p_{-1,x}(m),\\
u_{3,t_{m}}&=& u_{2}p_{-1}(m)-p_{-1}(m)u_{2}+p_{-2,x}(m) \\
           &=&[u_2,p_{-1}(m)]+p_{-2,x}(m), \\
u_{4,t_{m}} &=&-u_{2} p_{-1,x}(m)+ p_{-1}(m)u_{2,x}+u_{3}p_{-1}(m)\\
              & & {}- p_{-1}(m)u_{3} +u_{2}p_{-2}(m)- p_{-2}(m)u_{2}+
              p_{-3,x}(m),\\
          &=&-u_{2} p_{-1,x}(m)+p_{-1}(m)u_{2,x}+[u_3,p_{-1}(m)]+[u_2,p_{-2}(m)]+p_{-3,x}(m).
\end{eqnarray*}

It is not hard to observe that left-multiplication operator $``\longleftarrow"$ plays a key role to produce the
commutator in the  flow equations, which introduces the main difference  between  KP hierarchy and ncKP
hierarchy. For the flow equations (\ref{floweqwithp}), the drawback of this
expression is that it includes the coordinates
($p_{m-2}(m),p_{m-3}(m),p_{m-4}(m), \cdots$), which have to be calculated
independently. We shall give a general expression of $p_j(m)$ in the next section.

  Now we consider the so-called $n$-reduction, i.e. for some
fixed natural number $n$, we require
\begin{eqnarray}
L^{n}=L^{n}_{+}.
\end{eqnarray}

The condition (2.8) is equivalent  to requiring $p_{j}(n)=0$
 for $j<0$. Thanks to (2.7) we can recursively express
 all ordinates $u_{k}$ with $k>n$  in terms of
 ($u_2,u_3,\cdots,u_n$). Then, in the case of  $n$-reduction, only the
 first $n-1$ dynamical coordinates ($u_2,u_3,\cdots,u_n$) are
 independent, and the other two families  of coordinates also become finite: ($p_{0}(n),p_{1}(n),\cdots ,p_{n-2}(n)$) and ($q_{0}(n),q_{1}(n),\cdots ,q_{n-2}(n)$).
For any fixed $n$ we can insert
 $p_{j}(n)$ given by (2.7) (now depending only on
 ($u_2,u_3,\cdots,u_n$)) into (2.5), such that (2.5) defines a closed time-evolution system of ($u_2,u_3,\cdots,u_n$).

Thus we can write the Lax equations of $n$-reduction in
the finite matrix form
\begin{eqnarray}\label{Untmflow}
U(n)_{t_{m}}=K(n,m)=A(n)P(n,m)
\end{eqnarray}
where
\begin{eqnarray*}
 U(n)&=& \begin{pmatrix} u_2 \\ u_3 \\ \vdots \\ u_n  \end{pmatrix},\\
P(n,m)&=& \begin{pmatrix} p_{-1}(m) \\ p_{-2}(m)\\
\vdots \\p_{-n+1}(m) \end{pmatrix}, \\
A(n)&=& \begin{pmatrix}
A_{21} & 0 & 0 & \cdots & 0 \\
A_{31} & A_{32} & 0 & \cdots & 0 \\
\vdots & \vdots & \vdots & \ddots & \vdots\\
A_{n-1,1} & A_{n-1,2} & A_{n-1,3} & \cdots & 0\\
A_{n,1} & A_{n,2} & A_{n,3} & \cdots & A_{n,n-1}
\end{pmatrix} .
\end{eqnarray*}

\sectionnew{A more direct method to the flow equations of ncKP}
   We have obtained a terse formula of the flow equations (2.5)  in the preceding section, however, $p_{j}(m)$
are not readily expressed by ($u_2,u_3,\cdots$). In other words, it is not an explicit representation of the
flow equations of the ncKP hierarchy, and cannot concretely express the commutators of $\{u_i\}$ to show the
difference between KP hierarchy and ncKP hierarchy. We approach to the problem in another way, which is inspired by\cite{dajun}.

First of all, we have the following:
\begin{lemma} Set $m\geq 2$,and let all constants of integral be zero .
The coefficients ($p_{0}(m),p_{1}(m),\cdots
 ,p_{m-2}(m)$) of $L_{+}^{m}$ can be expressed as:
 \begin{eqnarray*}
p_{m-2}(m)&=& mu_2,\\
p_{m-3}(m)&=& mu_3+\binom{m}{2}u_{2,x},\\
 p_{m-k}(m)&=&
\sum_{l=1}^{k-1} \binom{m}{l}u_{k-l+1}^{(l-1)} \\ & &
{}+\int \sum_{s=2}^{k-2}\sum_{j=1}^{k-s-1}\Big(\binom{m-s}{j}p_{m-s}(m)u_{k-s-j+1}^{(j)}\\&
&
{}-(-1)^{k-s-j}\binom{k-s-1}{k-s-j}u_{j+1}p_{m-s}^{(k-s-j)}(m)\Big)\,dx\\
& &
{}+\int \sum_{s=2}^{k-1}\Big(p_{m-s}(m)u_{k-s+1}-u_{k-s+1}p_{m-s}(m)\Big)\,dx\;,k=4,\cdots,m.
\end{eqnarray*}
\end{lemma}
\begin{proof}
 Taking the non-negative part on both sides of Lax equations (\ref{Laxeq}), we have
$ [L_{+}^{m},L]_{+}=0$. It follows that
\begin{eqnarray} \label{calculatepjm}
[\partial, L^m_+]=[L^m_+,L_{-}]_+
\end{eqnarray}
Substituting $L_{+}^{m} = \sum_{j=0}^m p_{j}(m)\partial^{j} $ into
\begin{eqnarray*}
\sum_{k=2}^{m}p_{m-k,x}(m)\partial^{m-k} &=&
\sum_{k=2}^{m}\sum_{l=1}^{k-1}
\binom{m}{l}u_{k-l+1}^{(l)}\partial^{m-k}\\ & & {}+
\sum_{k=4}^{m}\sum_{s=2}^{k-2}\sum_{j=1}^{k-s-1}\Big(\binom{m-s}{j}p_{m-s}(m)u_{k-s-j+1}^{(j)}\\
& &
{}-(-1)^{k-s-j}\binom{k-s-1}{k-s-j}u_{j+1}p_{m-s}^{(k-s-j)}(m)\Big)\partial^{m-k}\\
& &
{}+\sum_{k=3}^{m}\sum_{s=2}^{k-1}(p_{m-s}(m)u_{k-s+1}-u_{k-s+1}p_{m-s}(m))\partial^{m-k}\;,
\end{eqnarray*}
comparing the coefficients of $\partial^{m-k}$ on two sides of the equation, and integrating them,
we obtain the expression of
($p_{0}(m),p_{1}(m),\cdots
 ,p_{m-2}(m)$).\\
\end{proof}
{\bf Remark 1.\ } We would like to stress that $p_j(m)$ given in lemma 3.1 are differential
polynomials of \{$u_2,u_3,\cdots ,u_{m-j}$\}, which are the same as Eq.(\ref{diffploy}). In other words, they are
local although their expressions include integration. This is important to assure that the flow equations
of the ncKP hierarchy are local in Section 4. Here we write  $\partial^{-1}u =\int u \,dx $. Since  $\partial \partial^{-1}=\partial^{-1}\partial=1$, we have $\int u_{x} \,dx=u$.
Take $m=4$ for example, we have
\begin{eqnarray*}
p_{2}(4)&=& 4u_2,\\
p_{1}(4)&=& 4u_3+6u_{2,x},\\
 p_{0}(4)&=& 4u_4+6u_{3,x}+4u_{2,xx}+\int (2p_{2}(4)u_{2,x}+u_2p_{2,x}(4))\,dx\\
          & & +\int (p_{2}(4)u_{3}-u_3p_{2}(4)+p_{1}(4)u_{2}-u_2p_{1}(4))\,dx)\,dx\\
          &=& 4u_4+6u_{3,x}+4u_{2,xx}+\int (8u_{2}u_{2,x}+4u_2u_{2,x}+6u_{2,x}u_2-6u_2u_{2,x})\,dx\\
          &=& 4u_4+6u_{3,x}+4u_{2,xx}+6u^{2}_{2}.
\end{eqnarray*}
{\bf  Remark 2. } Comparing with the results of the KP hierarchy\cite{dajun},  $p_{m-k}(m)$ in lemma 3.1
have additional parts given by commutators of $[ p_{m-s}(m),u_{k-s+1}]$.


We present some results of $L_{+}^{m}$ below:
\begin{eqnarray*}
L^{2}_{+} &=& \partial^{2}+2u_2 \;,\\
 L^{3}_{+} &=&
\partial^{3}+3u_2\partial+3u_3+3u_{2,x} \;,\\
 L^{4}_{+} &=&
\partial^{4}+4u_2\partial^{2}+(4u_3+6u_{2,x})\partial +4u_{4}+6u_{3,x}+4u_{2,xx} +6u^{2}_{2}\;,\\
 L^{5}_{+} &=&
\partial^{5}+5u_2\partial^{3}+(5u_3+10u_{2,x})\partial^{2} +(5u_{4}+10u_{3,x}+10u_{2,xx} +10u^{2}_{2})\partial \\
& &{}
+5u_{5}+10u_{4,x}+10u_{3,xx}+5u_{2,xxx}+20u_2u_3+10[u_3,u_2]+20u_2u_{2,x}+10[u_{2,x},u_2]\;,\\
L^{6}_{+} &=&
\partial^{6}+6u_2\partial^{4}+(6u_3+15u_{2,x})\partial^{3} +(6u_{4}+15u_{3,x}+20u_{2,xx}+15u^{2}_{2})\partial^{2}\\
& & {}
+(6u_{5}+15u_{4,x}+20u_{3,xx}+15u_{2,xxx}+15[u_2,u_3]+30u_3u_2+25[u_2,u_{2,x}]+45u_{2,x}u_2)\partial\;\\
& & {}
+6u_{6}+15u_{5,x}+20u_{4,xx}+15u_{3,xxx}+6u_{2,xxxx}+15u^{2}_3+25u^{2}_{2,x}+20u^{3}_2+20[u_2,u_{2,xx}]\\
& &
{}+35u_{2,xx}u_2+15[u_2,u_4]+30u_4u_2+20[u_{2,x},u_3]+30u_3u_{2,x}+25[u_2,u_{3,x}]+45u_{3,x}u_2\;.
\end{eqnarray*}

 Now we come to one of the main results.
\begin{theorem} Set $h=1,2,\cdots,m=1,2,\cdots$, then the flow equations of the ncKP hierarchy can be
explicitly expressed as
\begin{eqnarray}
u_{h+1,t_{m}}&=& \sum_{l=1}^{m}
\binom{m}{l}u_{m+h-l+1}^{(l)}+\sum_{s=2}^{m-1}\sum_{l=1}^{m-s}\binom{m-s}{l}p_{m-s}(m)u_{m+h-s-l+1}^{(l)} \nonumber\\
 & & {}-\sum_{s=2}^{m}\sum_{j=1}^{h+m-s-1}(-1)^{h+m-s-j}\binom{h+m-s-1}{h+m-s-j}u_{j+1}p_{m-s}^{(h+m-s-j)}(m)
 \nonumber \\
& & {}+\sum_{s=2}^{m}(p_{m-s}(m)u_{m-s+h+1}-u_{m-s+h+1}p_{m-s}(m)). \label{explictfloweqs}
\end{eqnarray}
\end{theorem}
\begin{proof}
Putting $L_{+}^{m} = \sum_{j=0}^m p_{j}(m)\partial^{j} $ back into the Lax equations, we obtain the explicit expression of flow equations (\ref{explictfloweqs}).
\end{proof}

{\bf Remark 3.} From the explicit expression of the flow equations
(\ref{explictfloweqs}) we can see the additional terms in the flow equations of the ncKP hierarchy are given by commutators of \{$u_i$\} compared with the commutative version\cite{dajun}. These additional commutators in $u_{h+1,t_{m}}$ come from $[p_{m-s}(m),u_{m-s+h+1}]$ and the commutators in $p_{m-s}(m)$. If $R$ is
commutative, the flow equations reduce to the familiar commutative case\cite{dajun}.

By the flow equations (\ref{explictfloweqs}), we can readily use computer algebra to generate higher-order
flow equations. We present some
flow equations calculated by Maple:
\begin{eqnarray}
u_{2,t_{2}}&=& 2u_{3,x}+u_{2,xx}\;,\label{computatekp1} \\
u_{3,t_{2}}&=& 2u_{4,x}+u_{3,xx}+2u_2u_{2,x}+2[u_{2},u_3]\;,\label{computatekp2}\\
u_{4,t_{2}}&=& 2u_{5,x}+u_{4,xx}+4u_3u_{2,x}-2u_2u_{2,xx}+2[u_{2},u_4]\;,\nonumber\\
u_{2,t_{3}}&=& 3u_{4,x}+3u_{3,xx}+u_{2,xxx}+6u_2u_{2,x}+3[u_{2,x},u_2]\;,\label{computatekp3}\\
u_{3,t_{3}}&=& 3u_{5,x}+3u_{4,xx}+u_{3,xxx}+6u_2u_{3,x}+6u_{2,x}u_3+3[u_3,u_{2,x}]+3[u_{2},u_4]\;,\nonumber\\
u_{4,t_{3}}&=& 3u_{6,x}+3u_{5,xx}+u_{4,xxx}+3u_2u_{4,x}+3[u_{2,x},u_4]+9u_4u_{2,x}+6u_3u_{3,x}\nonumber\\
 & & {}-3u_3u_{2,xx}-3u_2u_{3,xx}+3[u_2,u_5]+3[u_{3},u_4]\;,\nonumber\\
 u_{2,t_{4}}&=& 4u_{5,x}+6u_{4,xx}+4u_{3,xxx}+u_{2,xxxx}+12u_3u_{2,x}+6[u_{2,x},u_3]+6[u_2,u_{3,x}]+12u_{3,x}u_2\nonumber\\
 & & {}+4u_{2,xx}u_2+6u^{2}_{2,x}+2u_2u_{2,xx}\nonumber\;,\nonumber\\
 u_{2,t_{5}}&=& 5u_{6,x}+10u_{5,xx}+10u_{4,xxx}+5u_{3,xxxx}+u_{2,xxxxx}+10u_{2,xx}u_3+10u_{2}u_{2,x}u_2\nonumber\\
 & & {}+10u_2u_{2,xxx}+5[u_{2,xxx},u_2]+20u_{2,x}u_{2,xx}+10[u_{2,xx},u_{2,x}]+20u_{3,xx}u_{2}+10[u_2,u_{3,xx}]\nonumber\\
 & & {}+30u_{3,x}u_{2,x}+20[u_{2,x},u_{3,x}]+20u_{2,x}u^{2}_{2}+10[u^{2}_2,u_{2,x}]+20u_{3}u_{3,x}+10[u_{3,x},u_{3}]\nonumber\\
 & & {}+20u_{4}u_{2,x}+10[u_{2,x},u_{4}]+20u_{4,x}u_{2}+10[u_{2},u_{4,x}]\;. \label{5thflowkp}
\end{eqnarray}

    Let us compute the first non-trivial equation of the ncKP hierarchy as a nonlinear partial differential equation.
After setting $t_{2}=y\;,t_{3}=t\;,u_{x}=u_2$, and eliminating $u_3,u_4$ from the flows (\ref{computatekp1}),
(\ref{computatekp2}) and (\ref{computatekp3}), we get the ncKP equation
\begin{eqnarray}\label{nckpeq}
3u_{yy}-6[u_{x},u_{y}]=(4u_{t}-u_{xxx}-6u^{2}_{x})_{x}.
\end{eqnarray}
If $u$ is independent of $y$, $u=u_{x}$, the ncKP equation reduce to the ncKdV equation
\begin{eqnarray}\label{ncKdV}
4u_{t}=u_{xxx}+6uu_{x}+6u_{x}u\; .
\end{eqnarray}

 Alternatively, we can obtain the ncKdV equation from the 2-reduction. There is only one independent dynamical
 coordinate $u_2$ under the assumption of 2-reduction, and then
 \begin{eqnarray*}
u_3 &=& -\frac{1}{2}u_{2,x}\;,\\
u_4 &=& \frac{1}{4}u_{2,xx}-\frac{1}{2}u^{2}_2\;,\\
u_5 &=& -\frac{1}{8}u_{2,xxx}+u_2u_{2,x}+\frac{1}{2}u_{2,x}u_2\;,\\
u_6 &=& \frac{1}{16}u_{2,xxxx}-\frac{11}{8}u_{2,x}^{2}-\frac{11}{8}u_2u_{2,xx}-\frac{3}{8}u_{2,xx}u_2+\frac{1}{2}u_{2}^{3}\;,\\
    &\vdots  &
\end{eqnarray*}
Substituting $u_3,u_4$ into (\ref{computatekp3}) we get the ncKdV
equation (\ref{ncKdV})where $u=u_2$. Substituting $u_3,u_4,u_5,u_6$
into (\ref{5thflowkp}) we derive the fifth order ncKdV equation:
\begin{eqnarray} \label{5thncKdVeq}
u_{t}&=&\frac{1}{16}u_{xxxxx}+\frac{5}{8}uu_{xxx}+\frac{5}{8}uu_{xxx}+\frac{5}{4}u_{xx}u_{x}+\frac{5}{4}u_{x}u_{xx}
\nonumber \\
& &+\frac{5}{2}(u^{2}u_{x}+uu_{x}u+u_{x}u^{2})
\end{eqnarray}

It is hard to find a relation between the third-order ncKdV (\ref{ncKdV}) and the fifth-order ncKdV
(\ref{5thncKdVeq}). In Section 4 we shall obtain a recursion operator for the ncKdV hierarchy.
With it at hand, we shall be able to calculate the higher-order flows  of the ncKdV hierarchy including the ncKdV equation (\ref{ncKdV}), the fifth-order ncKdV equation(\ref{5thncKdVeq}) in an elegant way.

\sectionnew{Recursion Operator of noncommutative KP's flow equations }
The aim of this section is to obtain a recursion operator mapping the lower-order flow equations to higher-order flow equations under the general assumption of $n$-reduction, and a byproduct is the solution to the open problem proposed by
P.J. Olver and V.V. Sokolov\cite{olver}.

Inspired by but different from \cite{Oevel}, we shall use the identity $(L^{m})(L^{n})=(L^{n})(L^{m})$. That is the key to surmount the difficulty of noncommutativity.
The main result of this section is
\begin{theorem}
The Lax equations of the ncKP hierarchy under $n$-reduction in the matrix form (2.9)
possess a recursion operator $\Phi(n)$ such that
\begin{eqnarray}\label{recursion}
U(n)_{t_{m+kn}}&=& K(n,m+kn) \nonumber \\
               &=& \Phi(n)K(n,m+(k-1)n) \nonumber \\
               &=& \Phi^{2}(n)K(n,m+(k-2)n)\nonumber  \\
               & \vdots & \nonumber  \\
                &=& \Phi^{k}(n)K(n,m)=\Phi^{k}(n)U(n)_{t_{m}}, \label{recursionoperator}
\end{eqnarray} for all integers $k=0,1,2,\cdots$ and $m,n\in N$.
   \end{theorem}

In order to prove this, let us prove two lemmas at first:
\begin{lemma}
  There exist $(n-1)\times(n-1)$ matrix $S(n)$ and $(n-1)\times n$ matrix
$T(n)$ such that
 \begin{eqnarray}\label{Pmplusn1}
P(n,m+n)=S(n)P(n,m)+T(n)\begin{pmatrix}p_{-n}(m)
\\\vdots \\ p_{-2n+1}(m)\end{pmatrix}
\end{eqnarray}
   \end{lemma}
\begin{proof} 
We calculate the negative part of $(L^{m+n})_{-}=((L^{m})(L^{n}))_{-}$, which is given by
\begin{eqnarray*}
((L^{m})(L^{n}))_{-} &=& (\sum_{\alpha\leq-1}
p_{\alpha}(m)\partial^{\alpha}\sum_{j=0}^{n}
p_{j}(n)\partial^{j})_{-}\\
                      &=& (\sum_{\alpha\leq-1}\sum_{k=0}^{n}
p_{\alpha}(m)\partial^{\alpha}
p_{k}(n)\partial^{k})_{-}\\
                      &=& (\sum_{\alpha\leq-1}\sum_{\beta\geq
                      0}\sum_{k=0}^{n}\binom{\alpha}{\beta} p_{\alpha}(m) p_{k}^{(\beta)}(n)\partial^{\alpha+k-\beta})_{-}\\
&=&
\sum_{j\leq-1}\sum_{s=j+1}^{n}\sum_{k=max(0,s)}^{n}\binom{j-s}{k-s}
p_{j-s}(m)
p_{k}^{(k-s)}(n)\partial^{j}\\
\end{eqnarray*}
Comparing with $((L^{m+n}))_{-}= \sum_{j\leq -1}
p_{j}(m+n)\partial^{j}$ yields
\begin{eqnarray}\label{pjmplusn1}
p_{j}(m+n)=\sum_{s=j+1}^{n}C_{js}(n)p_{j-s}(m)\;,j\leq -1
\end{eqnarray} where 
\begin{equation}\label{nccijoperator1}
C_{js}(n)=\sum_{k=max(0,s)}^{n}\binom{j-s}{k-s} \stackrel{\longleftarrow}{p_{k}^{(k-s)}(n)};\,j\leq -1\;,s=j+1,j+2,\cdots,n.
\end{equation}

In particular,from $p_{n}(n)=1 \;, p_{n-1}(n)=0$ we have
\begin{eqnarray}\label{cjn1}
C_{jn}(n)=1\;,C_{j,n-1}(n)=0.\end{eqnarray} Introducing the
$(n-1)\times(n-1)$ matrix $S(n)$ and the $(n-1)\times n$ matrix
$T(n)$ $$S(n)=\begin{pmatrix}
C_{-1,0}(n) & C_{-1,1}(n) &  \cdots & C_{-1,n-2}(n)\\
C_{-2,-1}(n) & C_{-2,0}(n) & \cdots & C_{-2,n-3}(n)\\
\vdots & \vdots & \ddots & \vdots\\
C_{-n+1,-n+2}(n) & C_{-n+1,-n+3}(n)&  \cdots & C_{-n+1,0}(n)
\end{pmatrix} $$\\
$$T(n)=\begin{pmatrix}
C_{-1,n-1}(n) & C_{-1,n}(n) &  0 & \cdots & 0\\
C_{-2,n-2}(n) & C_{-2,n-1}(n) &  C_{-2,n}(n) & \cdots & 0\\
\vdots & \vdots & \ddots & \vdots\\
C_{-n+2,2}(n) & C_{-n+2,3}(n)&  C_{-n+2,4}(n)& \cdots & 0\\
C_{-n+1,1}(n) & C_{-n+1,2}(n)& C_{-n+1,3}(n)&  \cdots &
C_{-n+1,n}(n)
\end{pmatrix}. $$\\

Equations (\ref{pjmplusn1}) for $j=-1,\cdots,-n+1$ are written as
\begin{eqnarray*}\label{Pmplusn}
P(n,m+n)=S(n)P(n,m)+T(n)\begin{pmatrix}p_{-n}(m)
\\\vdots \\ p_{-2n+1}(m)\end{pmatrix}.
\end{eqnarray*}
\end{proof}
%
We now try to substitute the vector
$(p_{-n}(m),\cdots,p_{-2n+1}(m))^{T}$ in terms of the vector
$P(n,m)$ . This is
\begin{lemma}
  There exist $n\times n$ invertible matrix $M(n)$ and $n\times (n-1)$ matrix
$N(n)$ such that
 \begin{eqnarray}\label{MNrelation1}
\begin{pmatrix} p_{-n}(m)
\\p_{-n-1}(m)\\\vdots \\p_{-2n+2}(m)\\ p_{-2n+1}(m)
\end{pmatrix}=-M(n)^{-1}N(n)\begin{pmatrix} p_{-1}(m)
\\p_{-2}(m)\\\vdots \\p_{-n+2}(m)\\ p_{-n+1}(m)
\end{pmatrix}.
\end{eqnarray}
\end{lemma}

\begin{proof} Since $L^{m+n}=L^nL^m$, its negative part can be obtained by
\begin{eqnarray*}
((L^{n})(L^{m}))_{-} &=& (\sum_{\alpha=0}^{n}\sum_{\beta\leq-1}
p_{\alpha}(n)\partial^{\alpha}
p_{\beta}(m)\partial^{\beta})_{-}\\
                      &=& (\sum_{\alpha=0}^{n}\sum_{r=0}^{\alpha}\sum_{\beta\leq-1}\binom{\alpha}{r}
p_{\alpha}(n)
p_{\beta}^{(r)}(m)\partial^{\alpha+\beta-r})_{-}\\
                     &=&  (\sum_{s=0}^{n}\sum_{r=0}^{n-s}\sum_{\beta\leq-1}\binom{s+r}{r}
p_{s+r}(n)
p_{\beta}^{(r)}(m)\partial^{s+\beta})_{-}\\
&=& \sum_{j\leq-1}\sum_{s=0}^{n}\sum_{r=0}^{n-s}\binom{s+r}{r}
p_{s+r}(n)
p_{j-s}^{(r)}(m)\partial^{j}\\
\end{eqnarray*}
Comparing with $(L^{m+n})_{-}= \sum_{j\leq -1}
p_{j}(m+n)\partial^{j}$ yields
\begin{eqnarray}\label{pjmplusn2}
p_{j}(m+n)=\sum_{s=0}^{n}\tilde{C}_{s}(n)p_{j-s}(m)\;, j\leq -1\end{eqnarray}
where 
\begin{eqnarray}\label{nccijoperator2}
\tilde{C}_{s}(n)=\sum_{r=0}^{n-s}\binom{r+s}{r} p_{s+r}(n)\partial^{r}\;,0\leq s\leq n.
\end{eqnarray}
Again, from $p_{n}(n)=1 \;, p_{n-1}(n)=0$ we have, in particular,
\begin{eqnarray}\label{cjn2}
\tilde{C}_{n}(n)=1\;,\tilde{C}_{n-1}(n)=-n\partial\end{eqnarray}
Employing (\ref{pjmplusn1})=(\ref{pjmplusn2}) for $j=-1,\cdots,-n$ and observing (\ref{cjn1}) and
(\ref{cjn2}),  i.e.
\begin{eqnarray*}
\begin{pmatrix}
p_{-1}(m+n)
\\p_{-2}(m+n)\\\vdots \\p_{-n+1}(m+n)\\ p_{-n}(m+n)
\end{pmatrix} &=& \begin{pmatrix}
C_{-1,0}(n) & C_{-1,1}(n) &   \cdots & C_{-1,n-3}(n) & C_{-1,n-2}(n) \\
C_{-2,-1}(n) & C_{-2,0}(n) &   \cdots & C_{-2,n-4}(n) & C_{-2,n-3}(n) \\
\vdots & \vdots & \ddots & \vdots& \vdots\\
C_{-n+1,-n+2}(n) & C_{-n+1,-n+3}(n)&   \cdots & C_{-n+1,-1}(n)& C_{-n+1,0}(n)\\
C_{-n,-n+1}(n) & C_{-n,-n+2}(n)&  \cdots & C_{-n,-2}(n)&
C_{-n,-1}(n)\end{pmatrix}\\& &{}\times \begin{pmatrix} p_{-1}(m)
\\p_{-2}(m)\\\vdots \\p_{-n+2}(m)\\ p_{-n+1}(m)
\end{pmatrix}+ \begin{pmatrix}
0 & 0 &   \cdots & 0 & 0 \\
C_{-2,n-2}(n) & 0 &   \cdots & 0 & 0 \\
\vdots & \vdots & \ddots & \vdots& \vdots\\
C_{-n+1,1}(n) & C_{-n+1,2}(n)&   \cdots & 0& 0\\
C_{-n,0}(n) & C_{-n,1}(n)&  \cdots & C_{-n,n-2}(n)& 0\end{pmatrix}\\& &{}\times
\begin{pmatrix} p_{-n}(m)
\\p_{-n-1}(m)\\\vdots \\p_{-2n+2}(m)\\ p_{-2n+1}(m)
\end{pmatrix}+\begin{pmatrix} p_{-n-1}(m)
\\p_{-n-2}(m)\\\vdots \\p_{-2n+1}(m)\\ p_{-2n}(m)
\end{pmatrix}\\
&=& \begin{pmatrix}
\tilde{C}_{0}(n) & \tilde{C}_{1}(n) &   \cdots & \tilde{C}_{n-3}(n) & \tilde{C}_{n-2}(n) \\
0 & \tilde{C}_{0}(n)  &   \cdots & \tilde{C}_{n-4}(n)  & \tilde{C}_{n-3}(n)  \\
\vdots & \vdots & \ddots & \vdots & \vdots\\
0 & 0 &   \cdots & 0 & \tilde{C}_{0}(n)\\
0 & 0 &  \cdots & 0 & 0 \end{pmatrix}\times
\begin{pmatrix} p_{-1}(m)
\\p_{-2}(m)\\\vdots \\p_{-n+2}(m)\\ p_{-n+1}(m)
\end{pmatrix}\\ & &{}+ \begin{pmatrix}
-n\partial & 0 &   \cdots & 0 & 0 \\
\tilde{C}_{n-2}(n) & -n\partial &   \cdots & 0 & 0 \\
\vdots & \vdots & \ddots & \vdots & \vdots \\
\tilde{C}_{1}(n) & \tilde{C}_{2}(n) &   \cdots & -n\partial & 0\\
\tilde{C}_{0}(n) & \tilde{C}_{1}(n) &  \cdots & \tilde{C}_{n-2}(n) &
-n\partial  \end{pmatrix} \times
\begin{pmatrix} p_{-n}(m)
\\p_{-n-1}(m)\\\vdots \\p_{-2n+2}(m)\\ p_{-2n+1}(m)
\end{pmatrix}
\\ & &{}+\begin{pmatrix} p_{-n-1}(m)
\\p_{-n-2}(m)\\\vdots \\p_{-2n+1}(m)\\ p_{-2n}(m)
\end{pmatrix},
\end{eqnarray*}
we obtain
\begin{eqnarray}\label{MNrelation}
M(n)\begin{pmatrix} p_{-n}(m)
\\p_{-n-1}(m)\\\vdots \\p_{-2n+2}(m)\\ p_{-2n+1}(m)
\end{pmatrix}=-N(n)\begin{pmatrix} p_{-1}(m)
\\p_{-2}(m)\\\vdots \\p_{-n+2}(m)\\ p_{-n+1}(m)
\end{pmatrix}.
\end{eqnarray} Here the  $n\times n$ matrix $M(n)$ and the
$n\times(n-1)$ matrix $N(n)$ given by
\begin{eqnarray*}
M(n)&=& \begin{pmatrix}
-n\partial & 0 &   \cdots & 0 & 0 \\
D_{-2,n-2}(n) & -n\partial &   \cdots & 0 & 0 \\
\vdots & \vdots & \ddots & \vdots & \vdots \\
D_{-n+1,1}(n) & D_{-n+1,2}(n) &   \cdots & -n\partial & 0\\
D_{-n,0}(n) & D_{-n,1}(n) &  \cdots & D_{-n,n-2}(n) & -n\partial
\end{pmatrix}\\
N(n)&=& \begin{pmatrix}
D_{-1,0}(n) & D_{-1,1}(n) &   \cdots & D_{-1,n-3}(n) & D_{-1,n-2}(n) \\
C_{-2,-1}(n) & D_{-2,0}(n) &   \cdots & D_{-2,n-4}(n) & D_{-2,n-3}(n) \\
\vdots & \vdots & \ddots & \vdots& \vdots\\
C_{-n+1,-n+2}(n) & C_{-n+1,-n+3}(n)&   \cdots & D_{-n+1,-1}(n)& D_{-n+1,0}(n)\\
C_{-n,-n+1}(n) & C_{-n,-n+2}(n)&  \cdots & C_{-n,-2}(n)&
C_{-n,-1}(n)\end{pmatrix}
\end{eqnarray*}
with $$D_{j,s}=C_{j,s}-\tilde{C}_{s}(n).$$
 Note that entries of the
matrix $M(n)$ contain the left-multiplication operator, but $M(n)$
can be explicitly inverted in terms of the integration operator $\partial^{-1}$. By multiplying $M^{-1}$ on both sides of
(\ref{MNrelation}), we get (\ref{MNrelation1}).
\end{proof}
{\bf Proof of Theorem 4.1 :\  }
\begin{proof}
Taking  $(p_{-n}(m),\cdots,p_{-2n+1}(m))^{T}$ from
(\ref{MNrelation1}) back into (\ref{Pmplusn1}), we obtain a vital
recursion formula,
 \begin{eqnarray}\label{recursionf}
 P(n,m+n)=R(n)P(n,m)
\end{eqnarray}
 with
$$R(n)=S(n)-T(n)M(n)^{-1}N(n)\;.$$
Defining
$$\Phi(n):=A(n)R(n)A(n)^{-1}\;,$$
and  using  the definition of $U(n)$ and
 recursion formula (\ref{recursionf}), we obtain the recursion scheme of the flows equations under
 $n$-reduction,
\begin{eqnarray}
U(n)_{t_{m+kn}}&=& K(n,m+kn)=A(n)P(n,m+kn) \nonumber\\
                &=& A(n)R(n)P(n,m+(k-1)n)  \nonumber\\
               &=&  A(n)R(n)A(n)^{-1}A(n)P(n,m+(k-1)n)\nonumber\\
               &=&\Phi(n)A(n)P(n,m+(k-1)n)\nonumber  \\
                &=&  \Phi(n)K(n,m+(k-1)n)\nonumber\\
               &=& \Phi^{2}(n)K(n,m+(k-2)n)\nonumber\\
               & \vdots &\nonumber \\
                &=& \Phi^{k}(n)K(n,m) \nonumber
\end{eqnarray} for all integers $k=0,1,2,\cdots$ and $m,n\in N$.
\end{proof}

The integro-differential operator given by
the matrix $\Phi(n)$ is a recursion operator for Lax equations of the ncKP hierarchy under $n$-reduction. The flow equations commute in the noncommutative circumstances\cite{hamanaka2}, and
$\Phi(n)$ maps each flow equation in (\ref{Untmflow}) into a higher-order one.
Below let us see an example.

{\bf Example.\ } 2-reduction of the ncKP hierarchy generates the ncKdV hierarchy.

 The Lax operator is
 $$ L^{2}=L^{2}_{+}=\partial^{2}+2u $$
 For brevity we write $u=u_2$, then
\begin{eqnarray*}
A(2)&=&\partial  \\A(2)^{-1}&=& \partial^{-1}\\
S(2)&=&\stackrel{\leftarrow}{u}\\T(2)&=&(0,1)\\
M(2)&=&\begin{pmatrix}
-2\partial & 0  \\
2\stackrel{\leftarrow}{u}-2u-\partial^{2} & -2\partial
\end{pmatrix}\\
N(2)&=& \begin{pmatrix}
2\stackrel{\leftarrow}{u}-2u-\partial^{2}  \\
2\stackrel{\leftarrow}{u_{x}}
\end{pmatrix}\\
M(2)^{-1}&=& \begin{pmatrix}
-\frac{1}{2}\partial^{-1} & 0  \\
\frac{1}{4}-\frac{1}{2}\partial^{-1}(\stackrel{\leftarrow}{u}-2u)\partial^{-1}-\partial^{2}
& -\frac{1}{2}\partial^{-1}
\end{pmatrix}\\
R(2)&=& S(2)-T(2)M(2)^{-1}N(2)\\
     &=& \frac{1}{4}\partial^{2}+2\stackrel{\leftarrow}{u}-\partial^{-1}\stackrel{\leftarrow}{u_{x}}
-\frac{1}{2}[(\stackrel{\leftarrow}{u}-u)+\partial^{-1}(\stackrel{\leftarrow}{u}-u)\partial]\\
& &
{}+\partial^{-1}(\stackrel{\leftarrow}{u}-u)\partial^{-1}(\stackrel{\leftarrow}{u}-u)\\
\Phi(2)&=& A(2)R(2)A(2)^{-1}\\
   &=&\frac{1}{4}\partial^{2}+2\partial \stackrel{\leftarrow}{u}\partial^{-1}-\stackrel{\leftarrow}{u_{x}}\partial^{-1}
-\frac{1}{2}[\partial(\stackrel{\leftarrow}{u}-u)\partial^{-1}+(\stackrel{\leftarrow}{u}-u)]\\
& &
{}+(\stackrel{\leftarrow}{u}-u)\partial^{-1}(\stackrel{\leftarrow}{u}-u)\partial^{-1}\\
\end{eqnarray*}
Taking $k=m=1$ in (\ref{recursion}) we come back to the ncKdV equation:
\begin{eqnarray*}
u_{t}&=&\Phi(2)u_x=
\frac{1}{4}u_{xxx}+\frac{3}{2}(uu_{x}+u_{x}u)\;.
\end{eqnarray*}
Applying the recursion operator $\Phi(2)$ once again ,we return to
the fifth ncKdV equation:
\begin{eqnarray*}
u_{t}&=& \Phi(2)^2u_x= \Phi(2)(\frac{1}{4}u_{xxx}+\frac{3}{2}(uu_{x}+u_{x}u))\\
     &=&
    \frac{1}{16}u_{xxxxx}+\frac{5}{8}uu_{xxx}+\frac{5}{8}uu_{xxxx}+\frac{5}{4}u_{xx}u_{x}+\frac{5}{4}u_{x}u_{xx}+\frac{5}{2}(u^{2}u_{x}+uu_{x}u+u_{x}u^{2})\;.
\end{eqnarray*}

{\bf Remark 4.\ } This approach to get higher-order flows is more convenient than the method in
Section 3. After a rescaling we recover the recursion operator obtained in \cite{olver}. Since
our approach is based on the flow equations under $n$-reduction,
$\Phi(2)$ maps lower-order flows to higher-order flows. There is no doubt that $\Phi(2)$ produces a hierarchy of local, mutually commuting, higher-order flows, although $\Phi(2)$ is nonlocal itself. Actually the higher-order flows generated by $\Phi(2)$ from $u_x$  are
automatically local, they are the same as the flow equations
given by the Eq.(\ref{explictfloweqs}) under 2-reduction, since we obtain the recursion operator
$\Phi(2)$ from Lax
equations (\ref{Laxeq}), which are local. Thus we solve the open problem suggested by P.J. Olver and V.V. Sokolov
in \cite{olver}. We can elucidate it as follows. Our form of ncKdV equation
is
\begin{eqnarray*}
u_{t}= \frac{1}{4}u_{xxx}+\frac{3}{2}(uu_{x}+u_{x}u)\;.
\end{eqnarray*}
After the rescaling:
\begin{eqnarray*}
u &\rightarrow&2u\\
 x &\rightarrow& \frac{1}{2}x\\
t &\rightarrow& \frac{1}{2}t\;.
\end{eqnarray*}
we recover the form in \cite{olver}:
\begin{eqnarray*}
u_{t}= u_{xxx}+3(uu_{x}+u_{x}u)\;.
\end{eqnarray*}
The recursion operator in \cite{olver} is:
\begin{eqnarray*}
\Re=
D^{2}_{x}+2A_{u}+A_{u_{x}}D^{-1}_{x}+C_{u}D^{-1}_{x}C_{u}D^{-1}_{x}\;.
\end{eqnarray*}
where $D_{x}=\partial,A_u(v)=uv+vu,C_u(v)=[u,v].$

Then we have:
\begin{eqnarray*}
\Phi(2)v &=&
\frac{1}{4}v_{xx}+(uv+vu)+\frac{1}{2}(u_{x}(\partial^{-1}v)+(\partial^{-1}v)u_{x})+[u,\partial^{-1}[u,(\partial^{-1}v)]]\;,\\
\Re v &=&
v_{xx}+(uv+vu)+(u_{x}(\partial^{-1}v)+(\partial^{-1}v)u_{x})+[u,\partial^{-1}[u,(\partial^{-1}v)]]\;.
\end{eqnarray*}
The difference between $\Phi(2)v $ and $\Re v $ is exactly the scaling transformations above!

\sectionnew{Conclusion Remarks}
We have derived the  explicit expressions of the flow equations (\ref{explictfloweqs}) of  the ncKP hierarchy.
This result confirms a widely recognized conjecture that additional
terms in the flow equations of the ncKP hierarchy are given by
commutators of \{$u_i$\} compared with the commutative
version\cite{dajun}.
The recursion operator $\Phi(n)$ of the ncKP
hierarchy under $n$-reduction is given in Theorem 4.1.  For
2-reduction, $\Phi(2)$ is the recursion operator of the ncKdV hierarchy, which is consistent with the result of \cite{olver} by a scaling transformation. Here  $\Phi(2)$ generates automatically higher-order local flows from  $u_x$, thus we solve the open
problem proposed by P.J. Olver and V.V. Sokolov\cite{olver}.

W. Strampp and W. Oevel used the family of
coordinates $$(q_{m-2}(m),q_{m-3}(m),\cdots)$$ to express the flow equations of the classical KP hierarchy, and they obtained the corresponding recursion operator\cite{Oevel}. Why do we change the family of coordinates to $$(p_{m-2}(m),p_{m-3}(m),\cdots)?$$  This is one important feature of our paper. There are two
reasons: one is that we can express the flow equations of the ncKP hierarchy in a simpler fashion, the other reason is that we can derive the recursion operator of the ncKP hierarchy without the introduction of adjoint operator for the ncKP hierarchy, since the adjoint operator in the general noncommutative case has not been well defined as far as we know. It seems to us that due to the same obstacle there is no simple generalization of $\tau$ function and B-type hierarchy like Sato's theory \cite{sato} for the ncKP hierarchy.

\vspace{3ex}
{\bf Acknowledgments.} {\small This work is supported by NSF of China
under grant number 10971109 and K.C.Wong Magna Fund in Ningbo University.
He is also supported by the Program for NCET under Grant No.NCET-08-0515.
We sincerely thanks Professors Li Yishen and Cheng Yi(USTC, China) for their support and encouragement.}



\end{document}